\documentclass{llncs}
\usepackage{array}
\usepackage{amssymb}
\usepackage{amsfonts}
\usepackage{amsmath}
\usepackage{mathrsfs}
\usepackage{latexsym}
\usepackage{graphicx}
\usepackage{subfigure}
\usepackage{psfrag}
\usepackage{enumerate}
\usepackage{url}
\usepackage{color}
\usepackage{nicefrac}
\usepackage{verbatim}
\usepackage{authblk}
\usepackage[ruled,vlined]{algorithm2e}
\usepackage{textcomp}
\usepackage{tikz}
\usetikzlibrary{fit,calc,positioning,decorations.pathreplacing}

\def\mkvc{\textsc{max $k$-vertex cover}\xspace}

\def\sol{\mathrm{sol}}
\def\opt{\mathrm{opt}}
\def\SOL{\mathrm{SOL}}

\def\prf{\begin{proof}}
\def\eprf{\end{proof}}

\makeatletter
\def\@fnsymbol#1{\ensuremath{\ifcase#1\or (a)\or (b)\or (3)\or (4)\or \S \or * \or
   \mathsection\or \mathparagraph\or \|\or **\or \dagger \or \ddagger \or \dagger\dagger
   \or \ddagger\ddagger \else\@ctrerr\fi}}
\makeatother

\title{ 
A 0.821-ratio purely combinatorial algorithm for maximum $k$-vertex cover in bipartite graphs}
\author{Edouard Bonnet\inst{1,2} \and Bruno Escoffier\inst{4,5} \and Vangelis~Th.~Paschos\inst{1,2}\footnote{Institut Universitaire de France} \and Georgios Stamoulis\inst{1,2,3}}
\institute{PSL* Research University, Universit\'e Paris-Dauphine, LAMSADE \\
\email{\{edouard.bonnet,paschos\}@lamsade.dauphine.fr,georgios.stamoulis@dauphine.fr} \\
\and CNRS UMR 7243 \\
\and Universit\'{a} della svizzera Italiana \\
\and Sorbonne Universit\'es, UPMC Universite Paris 06 \\
\and UMR 7606, LIP6 \\
\email{bruno.escoffier@lip6.fr} \\
}

\protect

\protect

\begin{document}

\maketitle

\begin{abstract}

Our goal in this paper is to propose a \textit{combinatorial algorithm} that beats the only such algorithm known previously, the greedy one.
We study the polynomial approximation of \mkvc{} in bipartite graphs by a purely combinatorial algorithm
and present a computer assisted analysis of it, that finds the worst case approximation guarantee that is bounded below by~0.821.

\end{abstract}

\section{Introduction}\label{intro}

In the \mkvc{} problem, a graph~$G=(V,E)$ with $|V| = n$ and $|E| = m$ is given together with an integer $k \leqslant n$.
The goal is to find a subset $K \subseteq V$ with $k$ elements such that the total number of edges covered by~$K$ is maximized.
We say that an edge $e = \{u,v\}$ is covered by a subset of vertices~$K$ if $K \cap e \neq \emptyset$.
\mkvc{} is \textbf{NP}-hard in general graphs (as a generalization of \textsc{min vertex cover}) and it remains hard in bipartite graphs~\cite{apollonio14,DBLP:conf/ifipTCS/CaskurluMPS14}.

The approximation of \mkvc{} has been originally studied in~\cite{Hochbaum98}, where an approximation $1-\nicefrac{1}{e}$ was proved, achieved by the natural greedy algorithm. This ratio is tight even in bipartite graphs~\cite{DBLP:conf/compgeom/BadanidiyuruKL12}.
In~\cite{ageev}, using a sophisticated linear programming method, the approximation ratio for \mkvc{}
is improved up to~$\nicefrac{3}{4}$.
Finally, by an easy reduction from \textsc{Min Vertex Cover}, it can be shown that \mkvc{} can not admit a polynomial time approximation schema (PTAS), unless $\mathbf{P} = \mathbf{NP}$ \cite{DBLP:journals/cc/Patrank94}.

Obviously, the result of~\cite{ageev} immediately applies to the case of bipartite graphs.
Very recently,~\cite{DBLP:conf/ifipTCS/CaskurluMPS14} improves this ratio in bipartite graphs up to~$\nicefrac{8}{9}$, still using linear programming.

Finally, let us note that \mkvc{} is polynomial in regular bipartite graphs or in semi-regular ones, where the vertices of each color class have the same degree. Indeed, in both cases it suffices to chose~$k$ vertices in the color class of maximum degree.

\medskip
\noindent
\textbf{Our Contribution.} Our principal question motivating this paper is \textit{to what extent combinatorial methods for this problem compete with linear programming ones}. In other words, \textit{what is the ratios level, a purely combinatorial algorithm can guarantee?} In this purpose, we first devise a very simple algorithm that guarantees approximation ratio~$\nicefrac{2}{3}$, improving so the ratio of the greedy algorithm in bipartite graphs. Our main contribution consists of an approximation algorithm which computes six distinct solutions
and returns the best among them.

There is an obvious difficulty in analyzing the performance guarantee of such an algorithm.
Indeed it seems that there is no obvious way to compare different solutions and argue globally over them.
Another factor that contributes to this difficulty is that we provide analytic expressions for all the solutions produced, fact that involves a number of cases per each of them and a large number of variables (in all~48 variables are used for the several solution-expressions).
Similar situation was faced, for example, in~\cite{DBLP:journals/jal/FeigeKL02} where the authors gave a~$0.921$ approximation guarantee for \textsc{max cut} of maximal degree~3 (and an improved~$0.924$ for 3-regular graphs) by a computer assisted analysis of the quantities generated by theoretically analyzing a particular semi-definite relaxation of the problem at hand.
Similarly, by setting up a suitable non-linear program and solving it, we give a computer assisted analysis of a $0.821$-approximation guarantee for \mkvc{} in bipartite graphs. We give all the details of the implementation in Section~\ref{computeraided}.

\section{Preliminaries}\label{prelim}

The basic ideas of the algorithm we propose are the following: \\
\textbf{1.}~fix an optimal solution~$O$ (i.e., a vertex-set on~$k$ vertices covering a maximum number of edges in~$E$)  and guess the cardinalities~$k_1$ and~$k_2$ of its subsets~$O_1$ and~$O_2$ lying in the color-classes~$V_1$ and~$V_2$, respectively; \\
\textbf{2.}~compute the sets~$S_i$ of~$k_i$ vertices in~$V_i$, $i = 1,2$ that cover the most of edges; obviously~$S_i$ is a set of the~$k_i$ largest degree vertices in~$V_i$ (breaking ties arbitrarily); \\
\textbf{3.}~guess the cardinalities~$k'_i$ of the intersections $S_i \cap O_i$, $i = 1, 2$; \\
\textbf{4.}~compute the sets~$X_i$ of the $k_i - k'_i$ best vertices from~$V_i$ in graphs~$B[(V\setminus S_1), V_2]$ and~$B[V_1, (V_2\setminus S_2)]$, respectively; \\
\textbf{5.}~choose the best among six solutions built as described in Section~\ref{approxsection}.
%
%

Sets~$S_i$, $X_i$ and~$ O_i$ separate each color-class in~$6$ regions, namely, $S_i \cap O_i$, $S_i \setminus O_i$, $X_i \cap O_i$, $X_i \setminus O_i$, $O_i \setminus (S_i \cup X_i)$ (denoted by~$\bar{O}_i$, in what follows) and $V_i \setminus (S_i \cup X_i \cup O_i)$. So, there totally  exist~36 groups of edges (cuts) among them, the group $(V_1 \setminus (S_1 \cup X_1 \cup O_1), V_2 \setminus (S_2 \cup X_2 \cup O_2))$ being irrelevant as it will be hopefully understood in the sequel.
We will use the following notations to refer to the values of the 35 relevant cuts (illustrated in Figure~\ref{cuts}.):
\begin{description}
\item[$B$:]  the number of edges in the cut $(S_1 \setminus O_1, S_2 \cap O_2)$;
\item[$C$:]  the number of edges in the cut $(S_2 \setminus O_2, S_1 \cap O_1)$;
\item[$F_1,F_2,F_3$:] the number of edges in the cuts $(S_1 \setminus O_1, X_2 \setminus O_2)$, $(S_1 \setminus O_1, O_2 \setminus (X_2 \cup S_2))$ and $(S_1 \setminus O_1, O_2 \cap X_2)$, respectively;
\item[$H_1, H_2$:] the number of edges in the cuts $(S_1 \cap O_1, X_2 \setminus O_2)$ and $(S_1 \cap O_1, V_2 \setminus (S_2 \cup X_2 \cup O_2))$, respectively;
\item[$\{I_i\}_{i \in [6]}$:] the number of edges in the cuts $(X_1 \setminus O_1, X_2 \setminus O_2)$, $(X_1 \setminus O_1, V_2 \setminus (S_2 \cup X_2 \cup O_2))$, $(O_1 \setminus (S_1 \cup X_1), X_2 \setminus O_2)$, $(O_1 \setminus (S_1 \cup X_1), V_2 \setminus (S_2 \cup X_2 \cup O_2))$, $(X_1 \cap O_1, X_2 \setminus O_2)$ and $(X_1 \cap O_1, V_2 \setminus (S_2 \cup X_2 \cup O_2))$, respectively;
\item[$J_1,J_2,J_3$:] the number of edges in the cuts $({S_2} \setminus O_2, X_1 \setminus O_1)$, $({S_2} \setminus O_2, O_1 \setminus (S_1 \cup X_1))$ and
$({S_2} \setminus O_2, O_1 \cap X_1)$, respectively;
\item[$\{L_i\}_{i \in [9]}$:] the number of edges in the cuts $(S_1 \cap O_1, S_2 \cap O_2)$, $(S_1 \cap O_1, X_2 \cap O_2)$, $(S_1 \cap O_1, O_2 \setminus (S_2 \cup X_2))$, $(X_1 \cap O_1, S_2 \cap O_2)$, $(X_1 \cap O_1, X_2 \cap O_2)$, $(X_1 \cap O_1, O_2 \setminus (S_2 \cup X_2))$, $(O_1 \setminus (S_1 \cup X_1), S_2 \cap O_2)$, $(O_1 \setminus (S_1 \cup X_1), X_2 \cap O_2)$, and $(O_1 \setminus (S_1 \cup X_1), O_2 \setminus (S_2 \cup X_2))$, respectively;
\item[$N_1,N_2$:] the number of edges in the cuts $(S_2 \cap O_2, X_1 \setminus O_1)$ and $(S_2 \cap O_2, V_1 \setminus (S_1 \cup X_1 \cup O_1))$, respectively;
\item[$\{P_i\}_{i \in [5]}$:] the number of edges in the cuts $({X_2} \setminus O_2, V_1 \setminus (S_1 \cup X_1 \cup O_1))$, $({O_2} \setminus (S_2 \cup X_2), X_1 \setminus O_1)$, $({O_2} \setminus (S_2 \cup X_2), V_1 \setminus (S_1 \cup X_1 \cup O_1))$, $(X_2 \cap O_2, X_1 \setminus O_1)$, and $(X_2 \cap O_2, V_1 \setminus (S_1 \cup X_1 \cup O_1))$, respectively;
\item[$U_1, U_2, U_3$:] the number of edges is the cuts, $(S_1 \setminus O_1, S_2 \setminus O_2)$, $(S_1 \setminus O_1, V_2 \setminus (S_2 \cup X_2 \cup O_2))$ and $(S_2 \setminus O_2, V_1 \setminus (S_1 \cup X_1 \cup O_1))$, respectively.
\end{description}
\begin{figure}[h*]
\begin{center}
\includegraphics[scale=.55]{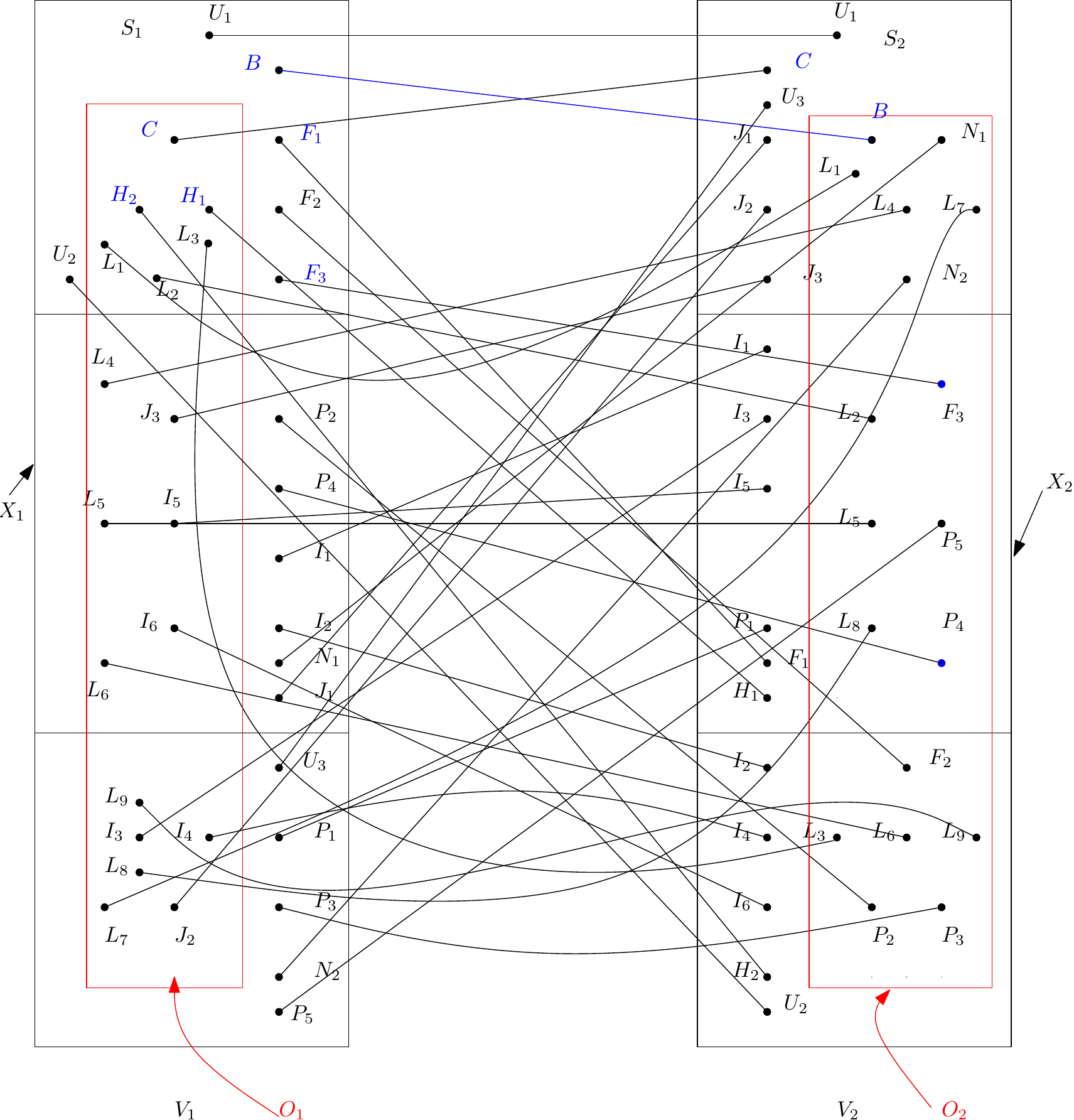}
\caption{Sets~$S_i$, $O_i$, $X_i$ $i= 1, 2$ and cuts between them.}\label{cuts}
\end{center}
\end{figure}
Based upon the notations above and denoting by~$\delta(V')$, $V' \subseteq V$, the number of edges covered by~$V'$ and by~$\opt(B)$ the value of an optimal solution (i.e., the number edges covered) for \mkvc{} in the input graph~$B$ , the following holds (see also Figure~\ref{cuts}):
\begin{eqnarray}
\delta\left(S_1\right) & = & B+C+F_1 + F_2 + F_3 + H_1 + H_2 + L_1 + L_2 + L_3 + U_1 + U_2 \label{s1} \\
\delta\left(S_2\right) & = & B+C+J_1 + J_2 + J_3 +L_1+L_4+L_7+ N_1 + N_2 + U_1 + U_3 \label{s2} \\
\delta\left(X_1\right) & = & I_1+I_2+I_5+I_6+J_1+J_3+ \sum_{i=4}^6L_i 
+N_1+P_2+P_4 \label{x1} \\
\delta\left(X_2\right) & = & F_1+F_3+H_1+I_1+I_3+I_5+L_2+L_5+L_8 
+P_1+P_4+P_5 \label{x2} \\ 
\delta\left(O_1\right) & = & C+H_1+H_2+I_3+I_4+I_5+I_6+J_2+J_3+ \sum_{i=1}^9L_i 
\label{o1} \\
\delta\left(O_2\right) & = & B+F_2+F_3+ \sum_{i=1}^9L_i 
+N_1+N_2+ \sum_{i=2}^5P_i\label{o2} \\ 
\opt(B) & = & 
B+C+ \sum_{i=2}^3F_i + \sum_{i=1}^2H_i + \sum_{i=3}^6I_i + \sum_{i=2}^3J_i + \sum_{i=1}^9L_i \nonumber \\
& & \mbox{} +\sum_{i=1}^2N_i  
+ \sum_{i=2}^5P_i \label{o} 
\end{eqnarray}
Without loss of generality, we assume $k_1 \leqslant k_2$ and we set:
$k_1 = \mu k_2$ ($\mu \leqslant 1$),
$k_1' = |S_1\cap O_1| = \nu k_1$ ($0 \leqslant \nu \leqslant 1$) and
$k_2' = |S_2\cap O_2| = \xi  k_2$ ($0 \leqslant \xi \leqslant 1$).
Let us note that, since~$k_i'$ vertices lie in the intersections $S_i \cap O_i$, the following hold for $\bar{O}_i = O_i \setminus (S_i \cup X_i)$, $i = 1,2$:
$|\bar{O}_1| = |O_1 \setminus (S_1 \cup X_1)| \leqslant (1-\nu)k_1 = \mu(1-\nu)k_2$ and $|\bar{O}_2| = |O_2 \setminus (S_2 \cup X_2)| \leqslant (1-\xi)k_2$.
%
From the definitions of the cuts and using~(\ref{s1}) to~(\ref{o2}) and the expressions for~$|\bar{O}_1|$ and~$|\bar{O}_2|$, simple average arguments and the assumptions for~$k_1$, $k_2$, $k_1'$ and~$k_2'$ just above, the following holds:
\begin{equation}\label{blockcons1}
\begin{array}{lcl}
\delta\left(S_1\right) &\geq& \delta\left(O_1\right) \\
\delta\left(S_2\right) &\geq& \delta\left(O_2\right) \\
\delta\left(X_1\right)+C+H_1+H_2+L_1+L_2+L_3 &\geq& \delta\left(O_1\right) \\
\delta\left(X_2\right) +B+N_1+N_2+L_1+L_4+L_7 &\geq& \delta\left(O_2\right) \\
\delta\left(S_1\right) &\geq& \nicefrac{1}{1-\nu}\cdot\delta\left(X_1\right) \\
\delta\left(S_2\right) &\geq& \nicefrac{1}{1-\xi}\cdot\delta\left(X_2\right) \\
\delta\left(S_1\right) + \delta\left(X_1\right) &\geq& \nicefrac{2-\nu}{1-\nu}\cdot\left(I_3+I_4+J_2+L_7+L_8+L_9\right) \\
\delta\left(S_2\right) + \delta\left(X_2\right) &\geq& \nicefrac{2-\xi}{1-\xi}\cdot\left(F_2+L_3+L_6+L_9+P_2+P_3\right) \\
B+F_1+F_2+F_3+U_1+U_2 &\geq& \delta\left(X_1\right) \\
C+J_1+J_2+J_3+U_1+U_3 &\geq& \delta\left(X_2\right)
\end{array}
\end{equation}
For $i= 1,2$, the two first inequalities in~(\ref{blockcons1}) hold because~$S_i$ is the set of~$k_i$ highest-degree vertices in~$V_i$; the third and fourth ones
because the lefthand side quantities are the number of edges covered by $X_i \cup (S_i\cap O_i)$; each of these sets has cardinality~$k_i$
and obviously covers more edges than~$O_i$; the fifth and sixth inequalities because the average degree of~$S_i$ is at least the average degree of~$X_i$ and $|X_1| = (1-\nu)k_1$ and $|X_2| = (1-\xi)k_2$; seventh and eighth ones because the average degree of vertices in $S_i \cup X_i$ is at least the average degree of vertices in $O_i \setminus (S_i \cup X_i)$; finally, for the last two inequalities the sum of degrees of the 
$k_i - k_i'$ vertices in~$S_i \setminus O_i$ is at least the sum of degrees of the $k_i - k_i'$ vertices of~$X_i$.

In Section~\ref{approxsection}, we specify the approximation algorithm sketched above. In Section~\ref{computeraided} a computer assisted analysis of its approximation-performance is presented. The non-linear program that we set up, not only computes the approximation ratio of our algorithm but it also provides an experimental study over families of graphs. Indeed, a particular configuration on the variables (i.e., a feasible value assignments on the variables that represent the set of edges $B,C,\dots$) corresponds to a particular family of bipartite graphs with similar structural properties (characterized by the number of edges belonging to the several cut considered). Given such a configuration, it is immediate to find the ratio of the algorithm, because we can simply substitute the values of the variables in the corresponding ratios and output the largest one. We can view our program as an \textit{experimental analysis} over all families of bipartite graphs, trying to find the particular family that implements the worst case for the approximation ratio of the algorithm. Our program not only finds such a configuration, but also provides data about the range of approximation factor on other families of bipartite graphs. Experimental results show that the approximation factor for the \textit{absolute majority} of the instances is very close to~1 i.e., $\geq 0.95$. Moreover, our program is \textit{independent} on the size of the instance. We just need a particular configuration on the relative value of the variables $B,C, \dots$, thus providing a compact way of representing families of bipartite graphs sharing common structural properties.

For the rest of the paper, we call ``best'' vertices a set of vertices that cover the most of \emph{uncovered} edges\footnote{For instance, saying ``we take~$S_1$ plus the~$k_2$ best vertices in~$V_2$, this means that we take~$S_1$ and then~$k_2$ vertices of highest degree in~$B[(V_1 \setminus S_1),V_2]$.} in~$B$.
Given a solution~$\SOL_k(B)$, we denote by~$\sol_k(B)$ its value. For the quantities implied in the ratios corresponding to these solutions, one can be referred to Figure~\ref{cuts}
and to expressions~(\ref{s1}) to~(\ref{o}).
Let us note that the algorithm above, since it runs for any value of~$k_1$ and $k_2$, it will run for $k_1 = k$ and $k_2 = k$. So, it is optimal for the instances of~\cite{DBLP:conf/compgeom/BadanidiyuruKL12}, where the greedy algorithm attains the ratio~$\nicefrac{(e-1)}{e}$.

Observe finally that, when $k \geqslant \min\{|V_1|,|V_2|\}$, then $\min\{|V_1|,|V_2|\}$ is an optimal solution since it covers the whole of~$E$. This remark will be useful for some solutions in the sequel, for example in the completion of solution~$\SOL_5(B)$.

\section{Some easy approximation results}

\subsection{A~$\mathbf{\nicefrac{2}{3}}$-approximation algorithm}\label{2/3}

The algorithm goes as follows: fix an optimal solution $O \subseteq V_1 \cup V_2$, guess~$k_1$ and~$k_2$, build the following three solutions and output the best among them:
\begin{itemize}
\item $\SOL_1$: take~$S_1$ plus the~$k_2$ remaining best vertices from~$V_2$;
\item $\SOL_2$: take~$S_2$ plus the~$k_1$ remaining best vertices from~$V_1$;
\item $\SOL_3$: take~$S_1$ plus~$S_2$.
\end{itemize}
$\SOL_1$ will cover more than $\delta(S_1) + \delta(O_2) - \delta(S_1,\bar{O}_2)$, where~$\bar{O}_2$ is $O_2 \setminus S_2$ and $\delta(S_1,\bar{O}_2)$ denotes the cardinality of the cut~$(S_1,\bar{O}_2)$. The fact that this solution covers more than $\delta(O_1)$ from the~$V_1$ side is obvious by the definition of~$S_1$. The~$k_2$ remaining best vertices from~$V_2$ will cover at least as many edges as $O \cap V_2$, except those that are already covered. This is precisely $\delta(O_2) - \delta(S_1, O_2)$ (we take something better than the ``surviving" part of~$O_2$).

With a complete analogy as for~$\SOL_1$, we have that~$\SOL_2$ will cover at least $\delta(S_2) + \delta(O_1) - \delta(S_2,\bar{O}_1) \geq \delta(O_2)  + \delta(O_1) - \delta(S_2,\bar{O}_1)$.

$\SOL_3$ will cover at least $\delta(S_1,O_2) \geq \delta(S1,\bar{O}_2)$ from~$V_1$. From~$S_2$ it will cover at least $\delta(S_2,\bar{O}_1) + \delta((S_2 \cap O_2),\bar{O}_1) \geqslant \delta(S_2,\bar{O}_1)$.

It is easy to see that $\sol_1(B) + \sol_2(B) + \sol_3(B) \geqslant 2(\delta(O_1) + \delta(O_2)) \geqslant 2\opt$, qed.

Let us note that that the algorithm above guarantees ratio~$\nicefrac{4}{5}$, when both $k'_i = 0$, $i = 1,2$~\cite{DBLP:journals/corr/BonnetEPS14}. Note also that, since it runs for any value of~$k_1$ and $k_2$, it will run for $k_1 = k$ and $k_2 = k$. So, it is optimal for the instances of~\cite{DBLP:conf/compgeom/BadanidiyuruKL12}, where the greedy algorithm attains the ratio~$\nicefrac{e-1}{e}$.

\subsection{The case $\nu = \xi =0$}\label{minor}

We present in this section a simple algorithm (Algorithm~\ref{alg2}) handling the case where $O_1 \cap S_1 = \emptyset$ and $O_2 \cap S_2 = \emptyset$ (notice that this case is not polynomially detectable). We show that in this case, a $\nicefrac{4}{5}$-approximation ratio can be achieved.

Consider the following algorithm:
\begin{enumerate}
\item {for} $i := 0$ {to}~$k$ {do}:
\begin{enumerate}
\item compute the set~$A_i$ (resp.,~$B'_i$) on~$i$ (resp., $k-i$) vertices of highest degrees in~$V_1$ (resp.,~$V_2$);
\item remove~$A_i$ (resp.~$B'_i$) from the graph, and compute the set~$A'_i$ (resp.,~$B_i$) on $k-i$ (resp.~$i$)  vertices of highest degrees in~$V_2$ (resp.,~$V_1$) in the surviving graph;
\item store the two solutions $(A_i\cup A'_i)$ and $(B_i\cup B'_i)$;
\end{enumerate}
\item  returnn {the best solution stored (denoted by~$\SOL(B)$).}
\end{enumerate}

We now prove that if $\nu = \xi =0$, then $\sol(B) \geqslant \nicefrac{4}{5}\cdot\opt(B)$.

Fix an optimal solution $O= O_1\cup O_2$ 
and consider the iteration of the algorithm with $i=k_1$. Set $A = A_i \cup A'_i$ and $B = B_i \cup B'_i$. Since the algorithm is symmetric, we can assume w.l.o.g. that $k_1\leqslant \nicefrac{k}{2}$. For some set $A \subseteq V$ denote by~$e(A)$ the number of edges covered by~$A$.

Once~$A_i$ has been taken, then the choice of~$A'_i$ is optimal among the possible sets of $k-i$ vertices in~$V_2$. Hence:
\begin{equation}\label{eqbr1}
\sol(B) \geqslant e(A) \geq e\left(A_i\cup O_2\right) =\delta\left(A_i\right)+\delta\left(O_2\right)-\delta\left(A_i,O_2\right)
\end{equation}
where~$\delta(A_i,O_2)$ denotes the set of edges having one endpoint in~$A_i$ and the other one in~$O_2$. Similarly,
\begin{equation}\label{eqbr2}
\sol(B) \geq e(A) \geq e\left(A_i\cup B'_i\right) \geq \delta\left(B'_i\right) + \delta\left(A_i,O_2\right)
\end{equation}
Now, consider the solution when $i=k$, i.e., when Algorithm~\ref{alg2} takes the set~$A_k$ of~$k$ best vertices in~$V_1$. Since $k_1 \leq \nicefrac{k}{2}$ and~$O_1$ and~$A_i$ are disjoint, it holds that:
\begin{equation}\label{eqbr3}
\sol(B) \geq e\left(A_k\right) \geq e\left(A_i\cup O_1\right)= \delta\left(A_i\right) + \delta\left(O_1\right)
\end{equation}
Now, sum up~(\ref{eqbr1}), (\ref{eqbr2}) and~(\ref{eqbr3}) with coefficients respectively~2, 2 and~1, respectively. Then:
$$
5\sol(B) \geq 4e(A) + e\left(A_k\right) \geq 3\delta\left(A_i\right) +\delta\left(O_1\right) +2\delta\left(B'_i\right) + 2\delta\left(O_2\right)
$$
Note that $\opt(B) \leq \delta(O_1)+\delta(O_2)$. The results follows since by the choice of~$A_i$ and~$B'_i$ we have $\delta(A_i) \geq \delta(O_1)$ and $\delta(B'_i) \geq \delta(O_2)$.

\section{A 0.821-approximation for the bipartite max $k$-ver\-tex co\-ver}\label{approxsection}

Consider the following algorithm for \mkvc{} (called \texttt{$k$-VC\_ALGO\-RI\-THM} in what follows) which guesses~$k_1$, $k_2$, $k_1'$ and~$k_2'$, builds several feasible solutions and, finally, returns the best among them.

Fix an optimal solution~$O$, guess the cardinalities~$k_1$ and~$k_2$ of~$O_1$ and~$O_2$ (swap these sets if necessary in order that $k_1 \leqslant k_2$), compute the sets~$S_i$ of~$k_i$ vertices in~$V_i$, $i = 1,2$, that cover the most of edges, 
guess the cardinalities~$k'_i$ of the intersections $S_i \cap O_i$, $i = 1, 2$, compute the sets~$X_i$ of~$k_i - k'_i$ best vertices in $V_i \setminus S_i$, $i = 1,2$ and
build the following \mkvc{}-solutions: \\
$\mathbf{SOL_1(B)}$ and~$\mathbf{SOL_2(B)}$,  take, respectively,~$S_1$ plus the~$k_2$ remaining best vertices from~$V_2$,  and~$S_2$ plus the~$k_1$ remaining best vertices from~$V_1$; \\
$\mathbf{SOL_3(B)}$ takes first $S_1 \cup X_1$ in the solution and  completes it with the $(1-\mu(1-\nu))k_2$ best vertices from~$V_2$; \\
$\mathbf{SOL_4(B)}$ takes~$S_2$ and completes it either with vertices from~$V_2$, or with vertices from both~$V_1$ and~$V_2$; \\
$\mathbf{SOL_5(B)}$ takes a $\pi$-fraction of the best vertices in~$S_1$ and~$X_1$, $\pi \in (0, \nicefrac{1}{2}]$; then, solution is completed with the $k_1 + k_2 - \pi(2k_1 - k_1')$
best vertices in~$V_2$; \\
$\mathbf{SOL_6(B)}$ takes a $\lambda$-fraction of the best vertices in~$S_2$ and~$X_2$, $\lambda \in (0, \nicefrac{(1+\mu)}{(2-\xi)}]$; then solution is completed with the $k_1 + k_2 - \lambda(2k_2 - k_2')$
best vertices in~$V_1$. 
Let us note that the values of~$\lambda$ and~$\pi$ are \textit{parameters that we can fix}.

In what follows, we analyze solutions~$\SOL_1(B) \ldots \SOL_6(B)$ computed by \texttt{$k$-VC\_ALGORITHM} and give analytical expressions for their ratios.

\subsection{Solution~$\mathbf{SOL_1(B)}$} \label{sol1}

The best $k_2$ vertices in~$V_2$, provided that~$S_1$ has already been chosen, cover at least the maximum of the following quantities:
$$
\begin{array}{llll}
\mathcal{A}_1 & = & J_1+J_2+J_3+L_4+L_7+N_1+N_2+U_3 & \text{by~$S_2$} \\
\mathcal{A}_2 & = & I_1+I_3+I_5+L_5+L_8+P_1+P_4+P_5 & \text{by~$X_2$} \\
\mathcal{A}_3 & = & L_4+L_5+L_6+L_7+L_8+L_9+N_1+N_2+P_2+P_3+P_4+P_5 & \text{by~$O_2$}
\end{array}
$$
So, the approximation ratio for~$\SOL_1(B)$ satisfies:
\begin{equation}\label{r1}
r_1 = \frac{\delta\left(S_1\right) + \max \Big\{ \mathcal{A}_1, \mathcal{A}_2, \mathcal{A}_3 \Big \}}{\opt(B)}
\end{equation}

\subsection{Solution~$\mathbf{SOL_2(B)}$} \label{sol2}
Analogously, the best $k_1$ vertices in~$V_1$, provided that~$S_2$ has already been chosen, cover at least the maximum of the following quantities:
$$
\begin{array}{llll}
\mathcal{B}_1 & = & H_1 +H_2+F_1+F_2+F_3+L_2+L_3+U_2 & \text{by~$S_1$} \\
\mathcal{B}_2 & = & I_1+I_2+I_5+I_6 + L_5+L_6+P_2+P_4 & \text{by~$X_1$} \\
\mathcal{B}_3 & = & H_1+H_2+I_3+I_4+I_5+I_6+L_2+L_3+L_5+L_6+L_8+L_9 & \text{by~$O_1$}
\end{array}
$$
So, the approximation ratio for~$\SOL_2(B)$ satisfies:
\begin{equation}\label{r2}
r_2 = \frac{\delta\left(S_2\right) + \max\Big\{ \mathcal{B}_1, \mathcal{B}_2, \mathcal{B}_3 \Big \} }{\opt(B)}
\end{equation}

\subsection{Solution~$\mathbf{SOL_3(B)}$}\label{sol3}
Taking first $S_1 \cup X_1$ in the solution, $k - (k_1 + k_1 - k_1') = k_1 + k_2 - 2k_1 + k_1' = k_2 - (k_1 - k_1') = (1-\mu(1-\nu))k_2$ vertices remain to be taken in~$V_2$. The best such vertices will cover at least the maximum of the following quantities:
\begin{align}
\mathcal{C}_1 = & (1-\mu(1-\nu))\left(J_2+N_2+L_7+U_3\right) \label{31}  \\
\mathcal{C}_2 = & \frac{1-\mu(1-\nu)}{2-\xi}\left(I_3 + J_2+L_7+L_8+N_2+P_1+P_5+U_3\right) \label{32} \\
\mathcal{C}_3 = & \frac{1-\mu(1-\nu)}{3-2\xi}\left(I_3+J_2+L_7+L_8+L_9+N_2+P_1+P_3+P_5+U_3\right) \label{33}
\end{align}
where~(\ref{31}) corresponds to a completion by the $(1-\mu(1-\nu))k_2$ best vertices of~$S_2$, (\ref{32}) corresponds to a completion by the $(1-\mu(1-\nu))k_2$ best vertices of $S_2 \cup X_2$, while~(\ref{33}) corresponds to a completion by the $(1-\mu(1-\nu))k_2$ best vertices of $S_2 \cup X_2 \cup \bar{O}_2$.
The denominator $3-2\xi$ in~(\ref{33})  is due to the fact that, using the expression for~$\bar{O}_2$, $|S_2 \cup X_2 \cup (O_2 \setminus (S_2 \cup X_2))| \leqslant (3-2\xi)k_2$. So, the approximation ratio for~$\SOL_3(B)$ is:
\begin{equation}\label{r3}
r_3 = \frac{\delta\left(S_1\right) + \delta\left(X_1\right) + \max \Big\{ \mathcal{C}_1, \mathcal{C}_2, \mathcal{C}_3 \Big\}}{\opt(B)}
\end{equation}

\subsection{Solution~$\mathbf{SOL_4(B)}$} \label{sol4}

Once $S_2$ taken in the solution, $k_1 = \mu k_2$ are still to be taken. Completion can be done in the following ways:
\begin{enumerate}
\item \textbf{if} $k_1 \leqslant k_2 - k_2'$, i.e., $\mu \leqslant 1 - \xi$, the best vertices taken for completion will cover at least either a~$\nicefrac{\mu}{1-\xi}$ fraction of edges incident to~$X_2$, or a~$\nicefrac{\mu}{2(1-\xi)}$ fraction of edges incident to $X_2 \cup \bar{O}_2$, i.e., at least~$\mathcal{M}_1$ edges, where~$\mathcal{M}_1$ is given by:
\begin{equation}\label{casemuleq1-xi}
\max\left\{\frac{\mu}{1-\xi}\delta\left(X_2\right), \frac{\mu}{2(1-\xi)}\left(\delta\left(X_2\right) + F_2+L_3+L_6+L_9+P_2+P_3\right)\right\}
\end{equation}
\item \textbf{else}, completion can be done by taking the whole set~$X_2$ and then the additional vertices taken:
\begin{enumerate}
\item either within the rest of~$V_2$ covering, in particular, a $\min\{1,\nicefrac{\mu-1+\xi}{|\bar{O}_2|}\} \geqslant \min\{1,\nicefrac{\mu-1+\xi}{1-\xi}\}$ fraction of edges incident to~$\bar{O}_2$ (quantity~$\mathcal{M}_2$ in~(\ref{casemugeq1-xi})),
\item or in~$S_1$ covering, in particular, a $\nicefrac{\mu-1+\xi}{\mu}$ fraction of uncovered edges incident to~$S_1$ (quantity~$\mathcal{M}_3$ in~(\ref{casemugeq1-xi})),
\item or in $S_1 \cup X_1$ covering, in particular, a $\nicefrac{\mu-1+\xi}{\mu(2-\nu)}$ fraction of uncovered edges incident to $S_1\cup X_1$ (quantity~$\mathcal{M}_4$ in~(\ref{casemugeq1-xi})),
\item or, finally, in $S_1 \cup X_1 \cup \bar{O}_1$ covering, in particular, a $\nicefrac{\mu-1+\xi}{\mu(3-2\nu)}$ fraction of uncovered edges incident to this vertex-set (quantity~$\mathcal{M}_5$ in~(\ref{casemugeq1-xi}));
\end{enumerate}
in any case such a completion will cover a number of edges that is at least the maximum of the following quantities:
\begin{equation}\label{casemugeq1-xi}
\begin{array}{rcl}
\mathcal{M}_2 & = & \min\left\{1, \frac{\mu-1+\xi}{1-\xi}\right\}\left(F_2 +L_3+L_6+L_9+ P_2 + P_3\right) \\ 
\mathcal{M}_3 & = & \frac{\mu-1+\xi}{\mu}\left(F_2 + H_2+L_3+U_2\right)  \\ 
\mathcal{M}_4 & = & \frac{\mu-1+\xi}{\mu(2-\nu)}\left(F_2+H_2+I_2+I_6+L_3+L_6+P_2+U_2\right)  \\ 
\mathcal{M}_5 & = & \frac{\mu-1+\xi}{\mu(3-2\nu)}\left(F_2+H_2+I_2+I_4+I_6+L_3+L_6+L_9+P_2+U_2\right) 
\end{array}
\end{equation}
\end{enumerate}
Using~(\ref{casemuleq1-xi}) and~(\ref{casemugeq1-xi}), the following holds for the approximation ratio of~$\SOL_4(B)$:
\begin{equation}\label{r4}
r_4 = \frac{\delta\left(S_2\right) + \left\{\begin{array}{ll}
\mathcal{M}_1 & \mu \leq 1-\xi \\
\delta\left(X_2\right) + \max\left\{\mathcal{M}_2, \mathcal{M}_3, \mathcal{M}_4, \mathcal{M}_5\right\} &  \mu \geq 1 - \xi
\end{array}\right.}{\opt(B)}
\end{equation}

\subsection{Vertical separations~-- solutions~$\mathbf{SOL_5(B)}$ and~$\mathbf{SOL_6(B)}$} \label{vertcuts}

For $i = 1,2$, given a vertex subset $V' \subseteq V_i$, we call \emph{vertical separation of~$V'$ with parameter~$c \in (0,\nicefrac{1}{2}]$}, a partition of~$V'$ into two subsets such that one of them contains a $c$-fraction of the best (highest degree) vertices of $V'$.  Then, the following easy claim holds for a vertical separation of $V' \cup V''$ with parameter~$c$.
\begin{claim}
Let $A(V')$ be a fraction~$c$ of the best  vertices in $V'$ and $A(V'')$ the same in $V''$. Then $\delta(A(V')) + \delta(A(V'')) \geq c \delta(V'\cup V'')$.
\end{claim}
\begin{proof}
Assume that in~$V'$ we have~$n'$ vertices. To form~$A(V')$ we take the $cn'$ vertices of~$V'$ with highest degree. The average degree of~$V'$ is~$\nicefrac{\delta(V')}{n'}$. The average degree of~$A(V')$ is~$\nicefrac{\delta(A(V'))}{cn'}$. But, from the selection of~$A(V')$ as the~$cn'$ vertices with highest degree, we have that $\nicefrac{\delta(A(V'))}{cn'} \geq \nicefrac{\delta(V')}{n'} \Rightarrow \delta(A(V')) \geq c \delta(V')$. Similarly for~$V''$, i.e., $\delta(A(V'')) \geq c \delta(V'')$.
\end{proof}
Solutions~$\SOL_5(B)$ and~$\SOL_6(B)$ are based upon vertical separations of $S_i \cup X_i$, $i= 1,2$, with parameters~$\pi$ and~$\lambda$, called $\pi$- and $\lambda$-vertical separations, respectively.

The idea behind vertical separation, is to handle the scenario when there is a ``tiny'' part of the solution (i.e. few in comparison to, let's say, $k_1$ vertices) that covers a large part of the solution and the ``completion'' of the solution done by the previous cases does not contribute more than a small fraction to the final solution. The vertical separation indeed tries to identify such a small part, and then continues the completion on the other side of the bipartition.

\subsubsection{Solution~$\mathbf{SOL_5(B)}$.}\label{sol5}

It consists of \emph{separating $S_1 \cup X_1$ with parameter $\pi \in (0, \nicefrac{1}{2}]$, of taking a~$\pi$ fraction of the best vertices of~$S_1$ and of~$X_1$ in the solution and of completing it with the adequate vertices from~$V_2$}. A $\pi$-vertical separation of $S_1 \cup X_1$ introduces in the solution $\pi\left(2k_1 - k_1'\right) = \pi(2-\nu)\mu k_2$  vertices of~$V_1$, which are to be completed with:
$$
k -  \pi(2-\nu)\mu k_2 = (1+\mu)k_2 -  \pi(2-\nu)\mu k_2 = (1-\mu(2\pi -1) + \mu\nu\pi)k_2
$$
vertices from~$V_2$. Observe that such a separation implies the cuts with corresponding cardinalities~$B$, $C$, $F_i$, $i = 1, 2, 3$, $H_1$, $H_2$, $I_1$, $I_2$, $I_5$, $I_6$, $J_1$, $J_3$, $L_j$, $j = 1, \ldots,6$, $N_1$, $P_2$, $P_4$, $U_1$ and~$U_2$. Let us group these cuts in the following way:
\begin{equation}\label{pigroups}
\begin{array}{lcl}
\Pi_1 & = & C+J_1+J_3+U_1 \\
\Pi_2 & = & B+L_1+L_4+N_1 \\
\Pi_3 & = & F_3+L_2+L_5+P_4 \\
\Pi_4 & = & I_1+I_5+F_1+H_1 \\
\Pi_5 & = & F_2+L_3+L_6+P_2 \\
\Pi_6 & = & I_2+I_6+H_2+U_2
\end{array}
\end{equation}
We may also notice that group~$\Pi_1$ refers to $S_2 \setminus O_2$,~$\Pi_2$ refers to~$S_2 \cap O_2$,~$\Pi_3$ to $X_2 \cap O_2$,~$\Pi_5$ to~$\bar{O}_2$ and~$\Pi_4$ to~$X_2 \setminus O_2$. Assume that a $\pi_i < 1$ fraction of each group~$\Pi_i$, $i= 1, \dots 6$ contributes in the~$\pi$ vertical separation of $S_1 \cup X_1$. Then, a $\pi$-vertical separation of $S_1 \cup X_1$ will contribute with a value:
\begin{equation}\label{firstpartpi}
\sum\limits_{i=1}^6\pi_i\Pi_i \geqslant \pi\sum\limits_{i=1}^6\Pi_i
\end{equation}
to~$\sol_5(B)$. We now distinguish two cases. \\
%
\textbf{Case 1:} $(1-\mu(2\pi -1) + \mu\nu\pi)k_2 \geqslant k_2$, i.e., $1-\mu(2\pi -1) + \mu\nu\pi \geqslant 1$. Then we have: \\
\textit{1.} $\mu (1- 2\pi) + \mu\nu\pi \leq 1-\xi$; then, the partial solution induced by the $\pi$-vertical separation will be completed in such a way that the contribution of the completion is at least equal to $\max\{Z_i, i = 1, \ldots, 5\}$, where: \\
$Z_1$ refers to~$S_2$ plus the best $(1-\mu(2\pi -1) + \mu\nu\pi)k_2 - k_2 = (\mu(1- 2\pi) + \mu\nu\pi)k_2$ vertices of~$O_2$ having a contribution of:
\begin{eqnarray}\label{z1}
Z_1 &=& \sum\limits_{i=1}^2\left(1-\pi_i\right)\Pi_i + \left(J_2 + L_7 + N_2 + U_3\right) + \frac{\mu(1-2\pi)+\mu\nu\pi}{1-\xi}\left[\left(1-\pi_3\right)\Pi_3 \right. \nonumber \\
& & \left. \mbox{} + \left(1-\pi_5\right)\Pi_5 + \left(L_8 + L_9 + P_3+P_5\right)\right]
\end{eqnarray}
$Z_2$ refers to~$S_2$ plus the best $(\mu(1- 2\pi) + \mu\nu\pi)k_2$ vertices of~$X_2$ having a contribution of:
\begin{align}\label{z2}
Z_2 =& \sum\limits_{i=1}^2\left(1-\pi_i\right)\Pi_i + \left(J_2 + L_7 + N_2 + U_3\right) \nonumber \\
 & \mbox{} + \frac{\mu(1-2\pi)+\mu\nu\pi}{1-\xi}\left[\sum\limits_{j=3}^4\left(1-\pi_i\right)\Pi_i +\left(I_3 + L_8  + P_1 +P_5\right)\right]
\end{align}
$Z_3$ and~$Z_4$ refer to the best $(1-\mu(2\pi -1) + \mu\nu\pi)k_2$ vertices of $S_2 \cup X_2$ and of $S_2 \cup O_2$ having, respectively, contributions:
\begin{eqnarray}
Z_3 &=& \frac{1-\mu(2\pi-1) + \mu\nu\pi}{2-\xi}\left[\sum\limits_{i=1}^4\left(1-\pi_i\right)\Pi_i \right. \nonumber \\
& & \mbox{} \left. + \left(I_3 + J_2 + L_7 + L_8 + N_2 + P_1 + P_5 + U_3\right)\right] \label{z3} \\
Z_4 &=& \frac{1-\mu(2\pi-1) + \mu\nu\pi}{2-\xi}\left[\sum\limits_{i=1}^3 \left(1-\pi_i\right)\Pi_i + \left(1-\pi_5\right)\Pi_5 \right. \nonumber \\
& & \mbox{} \left. + \left(J_2 + L_7 + L_8 + L_9 + N_2 + P_3 +P_5 +  U_3\right)\right] \label{z4}
\end{eqnarray}
$Z_5$ refers to the best $(1-\mu(2\pi -1) + \mu\nu\pi)k_2$ vertices of $S_2 \cup X_2 \cup \bar{O}_2$ having a contribution of:
\begin{eqnarray}\label{z5}
Z_5 &=& \frac{1-\mu(2\pi-1) + \mu \nu \pi}{3-2\xi}\left[\sum\limits_{i=1}^5\left(1-\pi_i\right)\Pi_i \right. \nonumber \\
& & \mbox{} \left. + \left(I_3 + J_2 + L_7 + L_8 + L_9 + N_2 + P_1 + P_3 + P_5 + U_3\right)\right]
\end{eqnarray}
\textit{2.} $\mu(1- 2\pi) + \mu\nu\pi \geq 1-\xi$; in this case, the partial solution induced by the $\pi$-vertical separation will be completed in such a way that the contribution of the completion is at least $\max\{\Theta_i, i = 1, \ldots, 3\}$, where: \\
$\Theta_1$ refers to $S_2 \cup X_2$ plus the best $(\mu(1- 2\pi) + \mu\nu\pi-(1-\xi))k_2$ vertices of~$\bar{O}_2$, all this having a contribution of:
\begin{eqnarray}\label{theta1}
\Theta_1 &=& \sum\limits_{i=1}^4\left(1-\pi_i\right)\Pi_i + \left(I_3 + J_2 + L_7 + L_8 + N_2+P_1+P_5+U_3\right) \nonumber \\
& & \mbox{} + \frac{\mu(1-2\pi)+\mu \nu \pi -(1-\xi)}{1-\xi}\left[\left(1-\pi_5\right)\Pi_5 + L_9 + P_3 \right]
\end{eqnarray}
$\Theta_2$ refers to $S_2 \cup O_2$ plus the best $(\mu(1- 2\pi) + \mu\nu\pi-(1-\xi))k_2$ vertices of $X_2 \setminus O_2$, all this having a contribution of:
\begin{align}\label{theta2}
\Theta_2 =& \sum\limits_{i=1}^3\left(1-\pi_i\right)\Pi_i \nonumber \\
 & \mbox{} + \left(1-\pi_5\right)\Pi_5 + \left(J_2 + L_7 + L_8 + L_9 + N_2+P_3+P_5+U_3\right) \nonumber \\
 & \mbox{} + \frac{\mu(1-2\pi)+\mu\nu\pi-(1-\xi)}{1-\xi}\left[\left(1-\pi_4\right)\Pi_4 + I_3+P_1\right]
\end{align}
$\Theta_3$ refers to the best  $(1-\mu(2\pi -1) + \mu\nu\pi)k_2$ vertices of $S_2 \cup X_2 \cup \bar{O}_2$ having a contribution of:
\begin{eqnarray}\label{theta3}
\Theta_3 &=& \frac{1-\mu(2\pi-1) + \mu \nu \pi}{3-2\xi}\left[\sum\limits_{i=1}^5\left(1-\pi_i\right)\Pi_i \right. \nonumber \\
& & \mbox{} \left. + \left(I_3 + J_2 + L_7 + L_8 + L_9 + N_2 + P_1 + P_3+P_5+U_3\right)\right]
\end{eqnarray}
\textbf{Case 2:} $1-\mu(2\pi -1) + \mu\nu\pi < 1$. The partial solution induced by the $\pi$-vertical separation will be completed in such a way that the contribution of the completion is at least equal to $\max\{\Phi_i, i = 1, \ldots, 5\}$, where: \\
$\Phi_1$ refers to the best $(1-\mu(2\pi -1) + \mu\nu\pi)k_2$ vertices in~$S_2$ with a contribution:
\begin{eqnarray}\label{phi1}
\Phi_1 &=& (1-\mu(2\pi-1) + \mu\nu\pi)\left[\sum\limits_{i=1}^2\left(1-\pi_i\right)\Pi_i + \left(J_2+L_7+N_2+U_3\right)\right]
\end{eqnarray}
$\Phi_2$ refers to the best $(1-\mu(2\pi -1) + \mu\nu\pi)k_2$ vertices in~$X_2$ with a contribution:
\begin{eqnarray}\label{phi2}
\Phi_2 &=& \frac{1-\mu(2\pi-1) + \mu\nu\pi}{1-\xi} \left[\sum\limits_{i=3}^4\left(1-\pi_i\right)\Pi_i + \left(I_3+L_8+P_1+P_5\right)\right]
\end{eqnarray}
$\Phi_3$ refers to the best $(1-\mu(2\pi -1) + \mu\nu\pi)k_2$ vertices in~$O_2$ with a contribution:
\begin{eqnarray}\label{phi3}
\Phi_3 &=& (1-\mu(2\pi-1) + \mu\nu\pi) \left[\sum\limits_{i=2}^3\left(1-\pi_i\right)\Pi_i + \left(1-\pi_5\right)\Pi_5 \right. \nonumber \\
& & \mbox{} \left. + \left(L_7+L_8+L_9+N_2+P_3+P_5\right)\right]
\end{eqnarray}
$\Phi_4$ refers to the best $(1-\mu(2\pi -1) + \mu\nu\pi)k_2$ vertices in~$S_2 \cup X_2$ with a contribution:
\begin{eqnarray}\label{phi4}
\Phi_4 &=& \frac{1-\mu(2\pi-1) + \mu\nu\pi}{2-\xi}\left[\sum\limits_{j=1}^4 \left(1-\pi_j\right)\Pi_j \right. \nonumber \\
& & \mbox{} \left. + \left(I_3+J_2+L_7+L_8+N_2+P_1+P_5+U_3\right)\right]
\end{eqnarray}
$\Phi_5$ refers to the best $(1-\mu(2\pi -1) + \mu\nu\pi)k_2$ vertices in~$S_2 \cup X_2 \cup \bar{O}_2$ with a contribution:
\begin{eqnarray}\label{phi5}
\Phi_5 &=& \frac{1-\mu(2\pi-1) + \mu\nu\pi}{3-2\xi}\left[\sum\limits_{j=1}^5 \left(1-\pi_j\right)\Pi_j \right. \nonumber \\
& & \mbox{} \left. + \left(I_3+J_2+L_7+L_8+L_9+N_2+P_1+P_3+P_5+U_3\right)\right]
\end{eqnarray}
Setting $Z^* = \max\{Z_i: i= 1, \ldots 5\}$, $\Theta^* = \max\{\Theta_i: i = 1,2,3\}$ and $\Phi^* = \max\{\Phi_i: i= 1, \ldots 5\}$, and putting~(\ref{pigroups}) and~(\ref{firstpartpi}) together with expressions~(\ref{z1}) to~(\ref{phi5}), we get for ratio~$r_5$:
\begin{equation}\label{r5}
\frac{\sum\limits_{i=1}^{6} \pi_i \Pi_i + \left\{\begin{array}{ll}
\left\{\begin{array}{ll}
Z^* & \text{if } \mu (1- 2\pi) + \mu \nu \pi \leq 1-\xi \\
\Theta^* & \text{if } \mu (1- 2\pi) + \mu \nu \pi \geq 1-\xi
\end{array}\right\} & \text{case: } 1-\mu(2\pi -1) + \mu \nu \pi \geq 1 \\
\Phi^* & \text{case: } 1-\mu(2\pi -1) + \mu \nu \pi < 1
\end{array}\right.}{\opt(B)}
\end{equation}

\subsubsection{Solution~$\mathbf{SOL_6(B)}$.}\label{sol6}

Symmetrically to~$\SOL_5(B)$, solution~$\SOL_6(B)$ consists of \emph{separating $S_2 \cup X_2$ with parameter $\lambda$, of taking a~$\lambda$ fraction of the best vertices of~$S_2$ and~$X_2$ in the solution and of completing it with the adequate vertices from~$V_1$}. Here, we need that:
$$
\lambda\left(k_2 + k_2 - k_2'\right) \leqslant k \Rightarrow \lambda(2-\xi)k_2 \leqslant (1+\mu)k_2 \Rightarrow \lambda \leqslant \frac{1+\mu}{2-\xi} \Rightarrow \lambda \in \left(\left. 0, \frac{1+\mu}{2-\xi} \right]\right.
$$
A $\lambda$-vertical separation of $S_2 \cup X_2$ introduces in the solution $\lambda(2-\xi)k_2$ vertices of~$V_2$, which are to be completed with:
$$
k -  \lambda(2-\xi)k_2 = (1+\mu)k_2 -  \lambda(2-\xi)k_2 = (1+\mu - \lambda(2-\xi))k_2
$$
vertices from~$V_1$.

Observe that such a separation implies the cuts with corresponding cardinalities~$B$, $C$, $F_1$, $F_3$, $H_1$, $I_1$, $I_3$, $I_5$, $J_i$, $i = 1, 2, 3$, $L_1$, $L_2$, $L_4$, $L_5$, $L_7$, $L_8$, $N_1$, $N_2$, $P_1$, $P_4$, $P_5$, $U_1$ and~$U_3$. We group these cuts in the following way:
\begin{equation}\label{lambdagroups}
\begin{array}{lcl}
\Lambda_1 & = & B+F_1+F_3+U_1 \\
\Lambda_2 & = & C+H_1+L_1+L_2 \\
\Lambda_3 & = & J_3+I_5+L_4+L_5 \\
\Lambda_4 & = & I_1+J_1+N_1+P_4 \\
\Lambda_5 & = & I_3+J_2+L_7+L_8 \\
\Lambda_6 & = & N_2+P_1+P_5+U_3
\end{array}
\end{equation}
Group~$\Lambda_1$ refers to $S_1 \setminus O_1$,~$\Lambda_2$ to~$S_1 \cap O_1$,~$\Lambda_3$ to $X_1 \cap O_1$,~$\Lambda_5$ to~$\bar{O}_1$ and~$\Lambda_4$ to~$X_1 \setminus O_1$. Assume, as previously, that a $\lambda_i < 1$ fraction of each group~$\Lambda_i$, $i= 1, \dots 6$ contributes in the~$\lambda$ vertical separation of $S_2 \cup X_2$. Then, a $\lambda$-vertical separation of $S_2 \cup X_2$ will contribute with a value:
\begin{equation}\label{firstpartlambda}
\sum\limits_{i=1}^6\lambda_i\Lambda_i \geqslant \lambda\sum\limits_{i=1}^6\Lambda_i
\end{equation}
to~$\sol_6(B)$. We again distinguish two cases.
\begin{enumerate}
\item $(1+\mu - \lambda(2-\xi))k_2 \geqslant \mu k_2$, i.e.,  $1+\mu - \lambda(2-\xi) \geqslant \mu$. Here we have the two following subcases:
\begin{enumerate}
\item $1-\lambda(2-\xi) \leq (1-\nu)\mu$; then, the partial solution induced by the $\lambda$-vertical separation will be completed in such a way that the contribution of the completion is at least equal to $\Upsilon^* = \max\{\Upsilon_i, i = 1, \ldots, 5\}$, where: \\
$\Upsilon_1$ refers to~$S_1$ plus the best $(1-\lambda(2-\xi))k_2$ vertices of~$X_1$ having a contribution of:
\begin{eqnarray}\label{upsilon1}
\Upsilon_1 &=& \sum\limits_{i=1}^2\left(1-\lambda_i\right)\Lambda_i + \left(H_2 + F_2 + L_3+U_2\right) \nonumber \\
& & \mbox{} + \frac{1-\lambda(2-\xi)}{\mu(1-\nu)}\left[\sum\limits_{i=3}^4\left(1-\lambda_i\right)\Lambda_i +\left(I_2 +I_6 + L_6 + P_2\right)\right]
\end{eqnarray}
$\Upsilon_2$ refers to~$S_1$ plus the best $(1-\lambda(2-\xi))k_2$ vertices of~$O_1$ having a contribution of:
\begin{eqnarray}\label{upsilon2}
\Upsilon_2 &=& \sum\limits_{i=1}^2\left(1-\lambda_i\right)\Lambda_i +\left (H_2 + F_2 + L_3+U_2\right) + \frac{1-\lambda(2-\xi)}{\mu(1-\nu)}\left[\left(1-\lambda_3\right)\Lambda_3 \right. \nonumber \\
& & \left. \mbox{} + \left(1-\lambda_5\right)\Lambda_5 + \left(I_4+I_6+L_6+L_9\right)\right]
\end{eqnarray}
$\Upsilon_3$ and~$\Upsilon_4$ refer to the best $(1+\mu - \lambda(2-\xi))k_2$  vertices of $S_1 \cup X_1$ and $S_1 \cup O_1$ having, respectively, contributions:
\begin{eqnarray}
\Upsilon_3 &=& \frac{\mu + 1 -\lambda(2-\xi)}{\mu(2-\nu)}\left[\sum\limits_{i=1}^4\left(1-\lambda_i\right)\Lambda_i \right. \nonumber \\
& & \left. \mbox{} + \left(F_2 + H_2 + I_2 + I_6 + L_3 + L_6 + P_2+U_2\right)\right] \label{upsilon3} \\
\Upsilon_4 &=& \frac{\mu + 1 -\lambda(2-\xi)}{\mu(2-\nu)}\left[\sum\limits_{i=1}^3\left(1-\lambda_i\right)\Lambda_i + \left(1-\lambda_5\right)\Lambda_5 \right. \nonumber \\
& & \left. \mbox{} + \left(F_2 + H_2 + I_4 + I_6 + L_3 + L_6 + L_9+U_2\right)\right] \label{upsilon4}
\end{eqnarray}
$\Upsilon_5$ refers to the best $(1+\mu - \lambda(2-\xi))k_2$ vertices of $S_1 \cup X_1 \cup \bar{O}_1$ having a contribution of:
\begin{eqnarray}\label{upsilon5}
\Upsilon_5 &=& \frac{\mu + 1 -\lambda(2-\xi)}{\mu(3-2\nu)}\left[\sum\limits_{j = 1}^5 \left(1-\lambda_j\right) \Lambda_j  \right. \nonumber \\
& & \left. \mbox{} + \left(F_2+H_2+I_2+I_4+I_6+ L_3 + L_6 + L_9 +P_2+U_2\right)\right]
\end{eqnarray}
\item $1-\lambda(2-\xi) \geq (1-\nu)\mu$; in this case, the partial solution induced by the $\lambda$-vertical separation will be completed in such a way that the contribution of the completion is at least $\Psi^* = \max\{\Psi_i, i = 1, \ldots, 3\}$, where: \\
$\Psi_1$ refers to $S_1 \cup X_1$ plus the best $(1 - \lambda(2-\xi) - (1-\nu))k_2$ vertices of~$\bar{O}_1$, all this having a contribution of:
\begin{eqnarray}\label{psi1}
\Psi_1 &=& \sum\limits_{j = 1}^4\left(1-\lambda_j\right) \Lambda_j + \left(F_2 + H_2 + I_2+I_6 + L_3 + L_6 + P_2+U_2\right)  \nonumber \\
 & & \mbox{} + \frac{1-\lambda(2-\xi) - \mu(1-\nu)}{\mu(1-\nu)}\left[\left(1-\lambda_5\right)\Lambda_5 +I_4+L_9\right]
\end{eqnarray}
$\Psi_2$ refers to $S_1 \cup O_1$ plus the best $(1 - \lambda(2-\xi) - (1-\nu))k_2$ vertices of $X_1 \setminus O_1$, all this having a contribution of:
\begin{eqnarray}\label{psi2}
\Psi_2 &=& \sum\limits_{j = 1}^3\left(1-\lambda_j\right) \Lambda_j + \left(1-\lambda_5\right)\Lambda_5 \nonumber \\
& & \mbox{} + \left(F_2+H_2+I_4+I_6+L_3+L_6+L_9+U_2\right) \nonumber \\
& & \mbox{} + \frac{1-\lambda(2-\xi) - \mu(1-\nu)}{\mu(1-\nu)}\left[\left(1-\lambda_4\right)\Lambda_4 +\left(I_2+P_2\right)\right]
\end{eqnarray}
$\Psi_3$ refers to the best $(\mu +1 - \lambda(2-\xi))k_2$ vertices of $S_1 \cup X_1 \cup \bar{O}_1$ having a contribution of:
\begin{eqnarray}\label{psi3}
\Psi_3 &=& \frac{\mu+1-\lambda(2-\xi)}{\mu(3-2\nu)}\left[\sum\limits_{j=1}^5\left(1-\lambda_j\right)\Lambda_j \right. \nonumber \\
 & & \left. \mbox{} + \left(F_2+H_2+I_2+I_4+I_6+L_3+L_6+L_9+P_2+U_2\right)\right]
\end{eqnarray}
\end{enumerate}
\item $1+\mu - \lambda(2-\xi) \leqslant \mu$. The partial solution induced by the $\lambda$-vertical separation will be completed in such a way that the contribution of the completion is at least equal to $\Omega^* = \max\{\Omega_i, i = 1, \ldots, 5\}$, where: \\
$\Omega_1$ refers to the best $(1+\mu - \lambda(2-\xi))k_2$ vertices in~$S_1$ with a contribution:
\begin{eqnarray}\label{omega1}
\Omega_1 &=& \frac{1+\mu-\lambda(2-\xi)}{\mu}\left[\sum\limits_{j=1}^2\left(1-\lambda_j\right)\Lambda_j + \left(F_2+H_2+L_3+U_2\right)\right]
\end{eqnarray}
$\Omega_2$ refers to the best $(1+\mu - \lambda(2-\xi))k_2$ vertices in~$X_1$ with a contribution:
\begin{eqnarray}\label{omega2}
\Omega_2 &=& \frac{1+\mu-\lambda(2-\xi)}{\mu}\left[\sum\limits_{j=3}^4\left(1-\lambda_j\right)\Lambda_j + \left(I_2+I_6+L_6+P_2\right)\right]
\end{eqnarray}
$\Omega_3$ refers to the best $(1+\mu - \lambda(2-\xi))k_2$ vertices in~$O_1$ with a contribution:
\begin{eqnarray}\label{omega3}
\Omega_3 &=& \frac{1+\mu-\lambda(2-\xi)}{\mu}\left[\sum\limits_{j=2}^3\left(1-\lambda_j\right)\Lambda_j + \left(1-\lambda_5\right)\Lambda_5 \right. \nonumber \\
 & & \left. \mbox{} +\left (H_2+I_4 + I_6+L_3+L_6+L_9\right)\right]
\end{eqnarray}
$\Omega_4$ refers to the best $(1+\mu - \lambda(2-\xi))k_2$ vertices in~$S_1 \cup X_1$ with a contribution:
\begin{eqnarray}\label{omega4}
\Omega_4 &=& \frac{1+\mu-\lambda(2-\xi)}{\mu(2-\nu)}\left[\sum\limits_{j=1}^4\left(1-\lambda_j\right)\Lambda_j \right. \nonumber \\
 & & \left. \mbox{} + \left(F_2+H_2+I_2+I_6+L_3+L_6+P_2+U_2\right)\right]
\end{eqnarray}
$\Omega_5$ refers to the best $(1+\mu - \lambda(2-\xi))k_2$ vertices in~$S_1 \cup X_1 \cup \bar{O}_1$ with a contribution:
\begin{eqnarray}\label{omega5}
\Omega_5 &=& \frac{1+\mu-\lambda(2-\xi)}{\mu(3-2\nu)}\left[\sum\limits_{j=1}^5\left(1-\lambda_j\right)\Lambda_j \right. \nonumber \\
 & & \left. \mbox{} + \left(F_2+H_2+I_2+I_4+I_6+L_3+L_6+L_9+P_2+U_2\right)\right]
\end{eqnarray}
\end{enumerate}
Putting~(\ref{lambdagroups}) and~(\ref{firstpartlambda})  together with expressions~(\ref{upsilon1}) to~(\ref{omega5}), we get:
\begin{equation}\label{r6}
r_{6} = \frac{\sum\limits_{i=1}^{6} \lambda_i \Lambda_i + \left\{\begin{array}{ll}
\left\{\begin{array}{ll}
\Upsilon^* & \text{if } 1-\lambda(2-\xi) \leq (1-\nu)\mu \\
\Psi^* & \text{if } 1-\lambda(2-\xi) > (1-\nu)\mu
\end{array}\right\} & \text{case: } \mu + 1 -\lambda(2-\xi) \geq \mu \\
\Omega^*  & \text{case: } \mu + 1 -\lambda(2-\xi) < \mu
\end{array}\right.}{\opt(B)}
\end{equation}



\section{Results} 

To analyze the performance guarantee of \texttt{$k$-VC\_ALGORITHM},  we set up a non-linear program and solved it to optimality. Here, we interpret the set of edges  $B,C, F_i, \dots$, as \textit{variables} , the expressions in~(\ref{blockcons1}) as \textit{constraints} and the \textit{objective function} is $\min r (\equiv \max_{j=1}^{6} r_j )$. In other words, we try to find a value assignments to the set of variables such that the maximum among all the six ratios defined is minimized. This value would give us the desired approximation guarantee of \texttt{$k$-VC\_ALGORITHM}.

Towards this goal, we set up a GRG (Generalized Reduced Gradient \cite{GRG1}) program. The reasons this method is selected are presented in Section~\ref{computeraided}, as well as a more detailed description of the implementation. GRG is a generalization of the classical \textit{Reduced Gradient} method~\cite{reduced_gradient} for solving (concave) quadratic problems so that it can handle higher degree polynomials and incorporate non-linear constraints.
Table~\ref{tablefinalres} in the following Section~\ref{computeraided} shows the results of the GRG program about the values of variables and quantities.  The values of ratios $r_1 \div r_6$ computed for them are the following: 
%
\begin{eqnarray*}
r_1 &=& 0.81806 \\
r_2 &=& 0.81797 \\
r_3 &=& 0.79280 \\
r_4 &=& 0.79657 \\
\mathbf{r_5} &=& \mathbf{0.82104} \\
r_6 &\mathbf{=}& 0.82103
\end{eqnarray*}
These results correspond to the cycle that outputs the \textit{minimum} value for the approximation factor and this is~0.821, given by solution $\SOL_5$.


\textit{Remark. As we note in Section~\ref{computeraided}, the~GRG solver does not guarantee the global optimal solution. The~0.821 guarantee is the minimum value that the solver returns after several runs from different initial starting points. However, successive re-executions of the algorithm, starting from this minimum value, were unable to find another point with smaller value. In each one of these successive re-runs, we tested the algorithm on~1000 random different starting points (which is greater than the estimation of the number of local minima) and the solver did not find value worse that the reported one.}

\section{A computer assisted analysis of the approximation ratio of \texttt{$k$-VC\_ALGORITHM}}\label{computeraided}

\subsection{Description of the method}

In this section we give details of the implementation of the solutions of the previous sections (as captured by the corresponding ratios) and we explain how these ratios guarantee a performance ratio of~$0.821$, i.e., that there is always a ratio among the ones described that is within a factor of 0.821 of the optimal solution value for the bipartite \mkvc{}.

Our strategy can be summarized as follows. We see the cardinalities of all cuts defined in Section~\ref{prelim} as \emph{variables}. These quantities represent how many edges go from one specific part of the bi-partition to any other given part of the other side of the bipartition. Counting these edges gives the value of the desired solution. By a proper scaling (i.e., by dividing every variable by the maximum among them) we guarantee that all these variables are in~$[0,1]$. Our goal is to find a particular configuration (which means a value assignment on the variables) such that the \emph{maximum} among all the different ratios that define the solutions of the previous section is as low as possible. This will give the performance guarantee.

This boils down to an optimization problem which can be, more formally, described as follows:
\begin{equation}\label{optprob}
\min r^* ~\mbox{ such that }~\max\limits_i \left\{ r_i \right\} \leqslant r^*
\end{equation}
Unfortunately, given the nature of the constraints captured by~(\ref{optprob}), this is \emph{not} a linear problem even though each variable appears as a monomial on the numerator and denominator of each constraint. This is because the numerators of~$r_3$~(\ref{r3}), $r_4$~(\ref{r4}), $r_5$~(\ref{r5}) and~$r_6$~(\ref{r6})
are polynomials of degree~3 or~4. Otherwise we could easily set up and solve to optimality this optimization problem, with our favorite linear solver.

To the best of our knowledge, there are no commercial solvers for solving polynomial optimization problems to find the \emph{global} optimal solution. All solvers for such polynomial systems stuck on local optima. The task then is to run the solver many times, with different starting points and different parameters, and to apply knowledge and intuition about the ``ballpark'' of the optimal solution value together with the respective configuration of the values of the variables, to be sure (given an error~$\epsilon$ unavoidable in such situations) that the optimal (or an almost optimal) solution of~(\ref{optprob}) is reached.

We note here that a promising although, as we will shortly argue, unsuccessful  approach would be to set up a \emph{Mathematica${}^{\circledR}$} program and would solve it exploiting the command \texttt{solve} which solve to optimality a system of polynomial equations using Gr\"{o}bner basis approach. Unfortunately, this is a solver that \emph{solves} a system of polynomial equations, and not an optimizer. In other words, given such a system as an input on the \texttt{solve} environment, this will either report that no feasible solution in the domain exists, or report a solution (value on the variables) that satisfy the system. Another, more serious, limitation is the following: we do not seek a configuration of the variable that satisfies all constraints (ratios). But we seek a configuration of minimum value such that there exists at least one constraint with value greater than the value of the configuration. In other words, if we look more carefully on the constraints, we see that these are of the form~$\min r^*$ s.t. $\exists r_i \geq r^*$. It is far from obvious how, and if, such a system could be set up on such solvers (in which some constraints might be ``violated'' i.e., be less than the target value of~$r^*$).

Another way to understand the above is to define the objective function value~$F$ of a given configuration (values)~$C$ for all the variables included. Given $C \in [0,1]^{X}$ where~$X$ is the set of variables, let~$r_i$ be the values of the ratios corresponding to the particular solutions. Then $F(C) = \max \{ r_i \}$. Our goal is to minimize this objective function value, i.e., to find a configuration on the variables such that~$F(C)$ is as small as possible. Observe that for a particular~$C$ it might very well be the case that all but one~$r_i$s are less than~$F(C)$. The objective value is given by the maximum value of all these ratios. This complexity of the objective function is precisely the reason why it is difficult to apply the \texttt{solve} environment. There are more complications that arise of technical nature (such as the use of conditions and cases), that will be discussed shortly.

\subsection{Selection of the optimizer}

So we have to settle with polynomial optimizers that may stuck on local optima and then, applying external knowledge and with the help of repetitive experiments, we try to reach a global optimal solution. For this reason we used two widely used polynomial (non-linear) solvers: The GRG (Generalized Reduced Gradient) solver  and the DEPS  (Differential Evolution and Particle Swarm Optimization) solver developed in SUN labs.

We will describe in more detail the GRG method and the technical details of the program we set up to achieve the $0.821$-approximation guarantee (The DEPS optimizer gave better results). The GRG method allows us to solve non-linear and even non-smooth problems. It has many different options that we exploit in our way to to find a global optimal solution. The GRG algorithm is the convex analog of the simplex method where we allow the constraints to be arbitrary nonlinear functions, and we also allow the variables to possibly have lower and upper bounds. It's general form is the following:
\begin{displaymath}
\begin{array}{cll}
\max\;\;  (\min ) & f\left(\vec{x}\right)  & \\
s.t. & h_i^T\left(\vec{x}\right)  = 0 & \forall i \in [m], \vec{L} \leq \vec{x} \leq \vec{U}
\end{array}
\end{displaymath}
where~$\vec{x}$ is the $n$-dimensional variable vector,~$h_i$ is the $i$-th constraint, and~$\vec{L}$, $\vec{U}$ are $n$-dimensional vectors representing lower and upper bounds of the variables. For simplicity we assume that~$\mathbf{h}$ is a matrix with~$m$ rows (the constraints) and~$n$ columns (variables) with rank~$m$ (i.e.,~$m$ linear independent constraints). The GRG method assumes that the set~$X$ of variables can be partitioned into two sets~$(\alpha, \beta)$ (let~$\vec{\alpha}$ and~$\vec{\beta}$ be the corresponding vectors) such that:
\begin{enumerate}
\item $\vec{\alpha}$ has dimension~$m$ and~$\vec{\beta}$ has dimension $n-m$;
\item the variables in~$\alpha$ strictly respect the given bounds represented by~$\vec{L}_{\alpha}$ and~$\vec{U}_{\alpha}$; in other words, $\forall x_i \in \alpha$, $\vec{L}_{x_i} \leq x_i \leq \vec{U}_{x_i}$.
\item $\nabla_{\alpha} h(\vec{x})$ is non-singular (invertible) at $X = (\alpha, \beta)$. From the Implicit Function Theorem, we know that for any given $\beta \subseteq X$, $\exists \alpha = X \setminus \beta$ such that $h(\vec{\alpha}, \vec{\beta}) = 0$. This immediately implies that $\nicefrac{\textrm{d}\vec{\alpha}}{\textrm{d} \vec{\beta}} = (\nabla_{\vec{\alpha}} h(\vec{x}))^{-1} \nabla_{\vec{\beta}} h(\vec{x})$.
\end{enumerate}
The main idea behind GRG  is to select the direction of the independent variables (which are the analog of the non-basic variables of the SIMPLEX method)~$\vec{\beta}$ to be the reduced gradient as follows:
\begin{displaymath}
\nabla_{\vec{\beta}} \left(f\left(\vec{x}\right) - y^T h\left(\vec{x}\right)\right),  \mbox{ where }  \vec{y} = \frac{\textrm{d} \vec{\alpha}}{\textrm{d}\vec{\beta}} = \left(\nabla_{\alpha} h\left(\vec{x}\right)\right)^{-1}\nabla_{\vec{\beta}} h\left(\vec{x}\right)
\end{displaymath}
Then, the step size is chosen and a correction procedure applied to return to the surface  $h(\vec{x})=0$. The intuition is fairly simple: if, for a given configuration of the values of the variables, a partial derivative has large absolute value, then the GRG would try to change the value of the variable appropriately and observe how its partial derivative changes. The goal is to arrive at a point where all partial derivatives are zero. This can happen to any local or global optimal point. In a few words, the GRG method is viewed as a sequence of steps through feasible points~$\vec{x}^j$ such that the final vector of this sequence satisfied the famous KKT conditions of optimality of non-linear systems.

In order to derive these conditions, we first take the Langrangean of the above problem:
\begin{displaymath}
\mathcal{L}\left(\vec{x}, \vec{\ell}\right) = f\left(\vec{x}\right) + \sum_{j \in [m]} \ell_j h^j\left(\vec{x}\right) - \sum_{i \in [n]} {L}_i \left(x_i - {L}_i\right) + \sum_{i \in [n]} {U}_i \left(x_i - {U}_i\right)
\end{displaymath}
At the optimum point~$\vec{x}^*$ the KKT conditions would yield that:
\begin{displaymath}
\nabla \mathcal{L} = \nabla f\left(\vec{x}^*\right) + \sum_{j \in [m]} \ell_j \nabla h^j\left(\vec{x}^*\right) - \left(\vec{L}-\vec{U}\right) = 0
\end{displaymath}
coupled with the standard constraints derived from the complementary slackness conditions. This is the stopping criterion of an iteration, meaning that we hit a local minimum.

As mentioned above, by setting the objective function value for a given configuration~$C$ on the variables~$X$ to be $F(C) = \max \{r_i \}$, our goal is to find a feasible~$C$ that minimizes~$F(C)$. An important thing here is to explain what we mean by ``feasible''. Typically, not every assignment of values to variables counts as feasible, because it might violate some obvious restrictions i.e., it might be the case that under a given assignment of values we have $\delta(S_1) \leq \delta(O_1)$ which is of course impossible (remember that~$S_1$ is the set of the~$k_1$ vertices of the \emph{highest} degree in~$V_1$ and so, by definition, they cover more edges than the vertices in the part of the optimum in~$V_1$). So, in order to complete our program, we couple it with  all the constraints from block~(\ref{blockcons1}):
$$
\begin{array}{cl}
\min & F(C) = \max\limits_{i=1}^{6}\left\{r_i\right\} \\
\mathrm{s.t.} & (\ref{blockcons1})
\end{array}
$$

\subsection{Implementation}

We set up a GRG program with the following details:
\begin{description}
\item[Variables.] We have one binary variable for each set of edges as depicted in Figure~\ref{cuts} plus $\pi_i, \lambda_i$, $i = 1, \ldots, 5$, plus $\mu, \nu, \xi$. Let~$X$ be this set of variables. We have $|X| = 48$.
\item[Parameters.] We note that in the $\pi$-fraction and in the $\lambda$-fraction of the solutions~$\SOL_5$, and~$\SOL_6$, the numbers~$\pi$ and~$\lambda$ are \emph{not} variables, but rather \textit{parameters} that we are free to choose. For the purpose of our experiments, we tried several different values for $\lambda, \pi$. 
In Table~\ref{table1}, we report results for various different choices of values for parameters~$\lambda$ and~$\pi$.
\item[Constraints.] Expression~(\ref{blockcons1}) in Section~\ref{prelim}.
\item[Further details.] In order to be certain about the optimality of the results, we employ a 2-step strategy. First, we apply a ``multistart" on the optimizer. Roughly speaking, the multistart works as follows. We provide a random seed to the optimizer, together with a parameter $X$, which is a positive integer. Then, we partition the feasible region of the variables (which is a subset of the $n$-dimensional hypercube $[0,1]^n$, $n = $ number of variables) into $X$ segments. The selection of $X$ feasible starting points inside the hypercube is done \textit{randomly}. We try to identify the local minimum in the neighborhood of each starting point. The output of the algorithm is the minimum among all these local minima. The intuition is simple: there might be several minima and by selecting randomly different starting points we significantly increase the chance to hit the global optimum. Typical size of $X$ in our experiments is 1000 (which is much greater than the number of different local optima in any case). In other words, after one "cycle" finish (hit of some local minimum) another running immediately starts from a different starting point chosen randomly (which is basically a feasible configuration of the variables).
%

We run the algorithm~100 independent times. Also, in each iteration, we start the first cycle at a different starting point by selecting a different random seed. The purpose of the random seed is to initiate the algorithm at a random point (feasible or not). This also means that the starting point of the other cycles would be also determined accordingly.
\item[Differencing method.] In order to numerically compute the partial derivative of a given configuration, we use the \emph{Central Differencing} method: in order to compute the derivative we use two different configurations on the variables, in the opposite direction of each other, as opposed to the method of forward differencing which uses a single point  that is slightly different from the current point to compute the derivative.
In more detail, in order to compute the first derivative at point $x_0 \in [0,1]^n$ we use the following (where $h$ is the precision, or the ``spacing": typical values of $h$ in our applications are $< 0,00001$):
\begin{displaymath}
\partial_c f\left(x_0\right) = f\left(x_0 + \frac{1}{2}h\right) - f\left(x_0 - \frac{1}{2}h\right)
\end{displaymath}
The central differencing method we used, although more time-consuming since it needs more calculations, is more accurate since, when~$f$ is twice differentiable, the term~$\partial_c f(x_0)$ divided by the precision~$h$, incurs an error of~${O}(h^2)$ as opposed to error~${O}(h)$ that we would have if we were using forward (or backward) differencing. Of course this comes at a cost of time consumption reflected by the more calculations needed to approximate the derivatives, but precision is more important than time in our application.
\end{description}


\subsection{Results}

In this section we report the results of the GRG program. First, we summarize the results according to the different values of parameters~$\pi$ and~$\lambda$. One can see that as these values decrease, the approximation guarantee increases. Also, for convenience, we include the approximation guarantee returned by including only the four first rations (excluding $\SOL_5, \SOL_6$ corresponding to the two vertical cuts on~$V_1$ and~$V_2$ respectively; first line in Table~\ref{table1}).

\begin{table}[th]
\begin{center}
\begin{tabular}{c|c|c}
Value of $\pi$  &   Value of $\lambda$  &   Ratio \\
\hline
\hline
-       &   -   & 0.723269 \\
0.4     &   0.4 & 0.754895 \\
0.2     &   0.00001 &   0.776595 \\
0.1     &   0.1 & 0.780161 \\
0.05    &   0.1 & 0.795602 \\
0.0001  &   0.5 & 0.807453 \\
0.0001  &   0.0001  &   0.805927 \\
0.00001 &   0.00001 &   0.821044 \\
\hline
\hline
\end{tabular}
\caption{Results according to the different values of parameters~$\pi$ and~$\lambda$.}\label{table1}
\end{center}
\end{table}
In Table~\ref{tablefinalres}, the final results with~$\pi = \lambda = 10^{-5}$ are given.

\begin{table}[th]
\begin{center}
\begin{tabular}{c|r||c|r||c|r||c|r}
Variables & Values & Groups & Values & $\pi, \lambda$ & Values & Ratios & Values \\
\hline
\hline
$B$     &   1      & $\delta(S_1)$ & 5.28490 & $\pi$    & 0.00001 & $r_1$ & 0.81806 \\
$C$     &   0.9944 & $\delta(S_2)$ & 5.90033 & $\pi_1$  & 0.08471 & $r_2$ & 0.81797 \\
$F1$    &	0.0002 & $\delta(X_1)$ & 2.78398 & $\pi_2$  & 0.13072 & $r_3$ & 0.79280 \\
$F2$    &	0.4954 & $\delta(X_2)$ & 3.09961 & $\pi_3$  & 0.97865 & $r_4$ & 0.79657 \\
$F3$	&   0.4457 & $\delta(O_1)$ & 5.26489 & $\pi_4$  & 0.19364 & $r_5$ & \textbf{0.82104} \\
$H1$    &	0.8449 & $\delta(O_2)$ & 5.88331 & $\pi_5$  & 0.38861 & $r_6$ & 0.82103 \\
$H2$    &	0.0623 & $\delta(OPT)$ & 10.5589 &          &  &  & \\
$I1$    &	0      &   &   & $\lambda$   &  0.00001   &  & \\
$I2$    &	0      &   &   & $\lambda_1$ &  0.14995   &  &  \\
$I3$    &	0.9986 &   &   & $\lambda_2$ &  0.76660   &  &  \\
$I4$    &	0      &   &   & $\lambda_3$ &  0.15362   &  & \\
$I5$    &	0.0577 &   &   & $\lambda_4$ &  1         &  &   \\
$I6$    &	0.3740 &   &   & $\lambda_5$ &  1         &  &  \\
$J1$    &	0.2386 &   &   &             &                & & \\
$J2$    &	0.9824 &   &   &             &                & & \\
$J3$    &	0.3612 &   &   &             &                & & \\
$N1$    &	1      &   &   &             &                & & \\
$N2$    &	0.6005 &   &   &             &                & & \\
$P1$    &	0      &   &   &             &                & & \\
$P2$    &	0      &   &   &             &                & & \\
$P3$    &	1      &   &   &             &                & & \\
$P4$    &	0.7525 &   &   &             &                & & \\
$P5$    &	0      &   &   &             &                & & \\
$L1$    &	0.1932 &   &   &           &                  & & \\
$L2$    &	0      &   &   &           &                  & & \\
$L3$    &	0.3960 &   &   &           &                  & & \\
$L4$    &	0      &   &   &           &                  & & \\
$L5$    &	0      &   &   &           &                  & & \\
$L6$    &	0      &   &   &           &                  & & \\
$L7$    &	0      &   &   &           &                  & & \\
$L8$    &	0      &   &   &           &                  & & \\
$L9$    &	0      &   &   &           &                  & & \\
$U1$    &	0.5330 &   &   &           &                  & & \\
$U2$    &	0.3198 &   &   &           &                  & & \\
$U3$    &	0      &   &   &           &                  & & \\
$\mu$   &	0.809  &   &   &           &                  & & \\
$\nu$   &	0           &   &   &           &                  & & \\
$\xi$   &	0           &   &   &           &                  & & \\
\hline
\hline
\end{tabular}
\caption{The final results with~$\pi = \lambda = 10^{-5}$.}\label{tablefinalres}
\end{center}
\end{table}

Let us conclude noticing that the non-linear program that we set up, not only computes the approximation ratio of \texttt{$k$-VC\_ALGORITHM} but it also provides an experimental study over families of graphs. Indeed, a particular configuration on the variables (i.e., a feasible value assignments on the variables that represent the set of edges $B,C,\dots$) corresponds to a particular family of bipartite graphs with similar structural properties (characterized by the number of edges belonging to the several cut considered). Given such a configuration, it is immediate to find the ratio of \texttt{$k$-VC\_ALGORITHM}, because we can simply substitute the values of the variables in the corresponding ratios and output the largest one. We can view our program as an \textit{experimental analysis} over all families of bipartite graphs, trying to find the particular family that implements the worst case for the approximation ratio of the algorithm. Our program not only finds such a configuration, but also provides data about the range of approximation factor on other families of bipartite graphs. Experimental results show that the approximation factor for the \textit{absolute majority} of the instances is very close to~1 i.e., $\geq 0.95$. Moreover, our program is \textit{independent} on the size of the instance. We just need a particular configuration on the relative value of the variables $B,C, \dots$, thus providing a compact way of representing families of bipartite graphs sharing common structural properties.

We run the program on a standard~$C++$ implementation of the GRG algorithm on a 64-bit Intel Core i7-3720QM@2.6GHz, with 16GB of RAM at 1600MHz running Windows 7 x64 \textit{and} Ubuntu 9.10~x32. 

\medskip
\noindent
\textbf{Acknowledgement.} The work of the author was supported by the Swiss National Science Foundation Early Post-Doc mobility grand P1TIP2\_152282.

\newpage


%


\begin{figure}[h*]
\begin{center}
\includegraphics[width=\textwidth,height=\textheight]{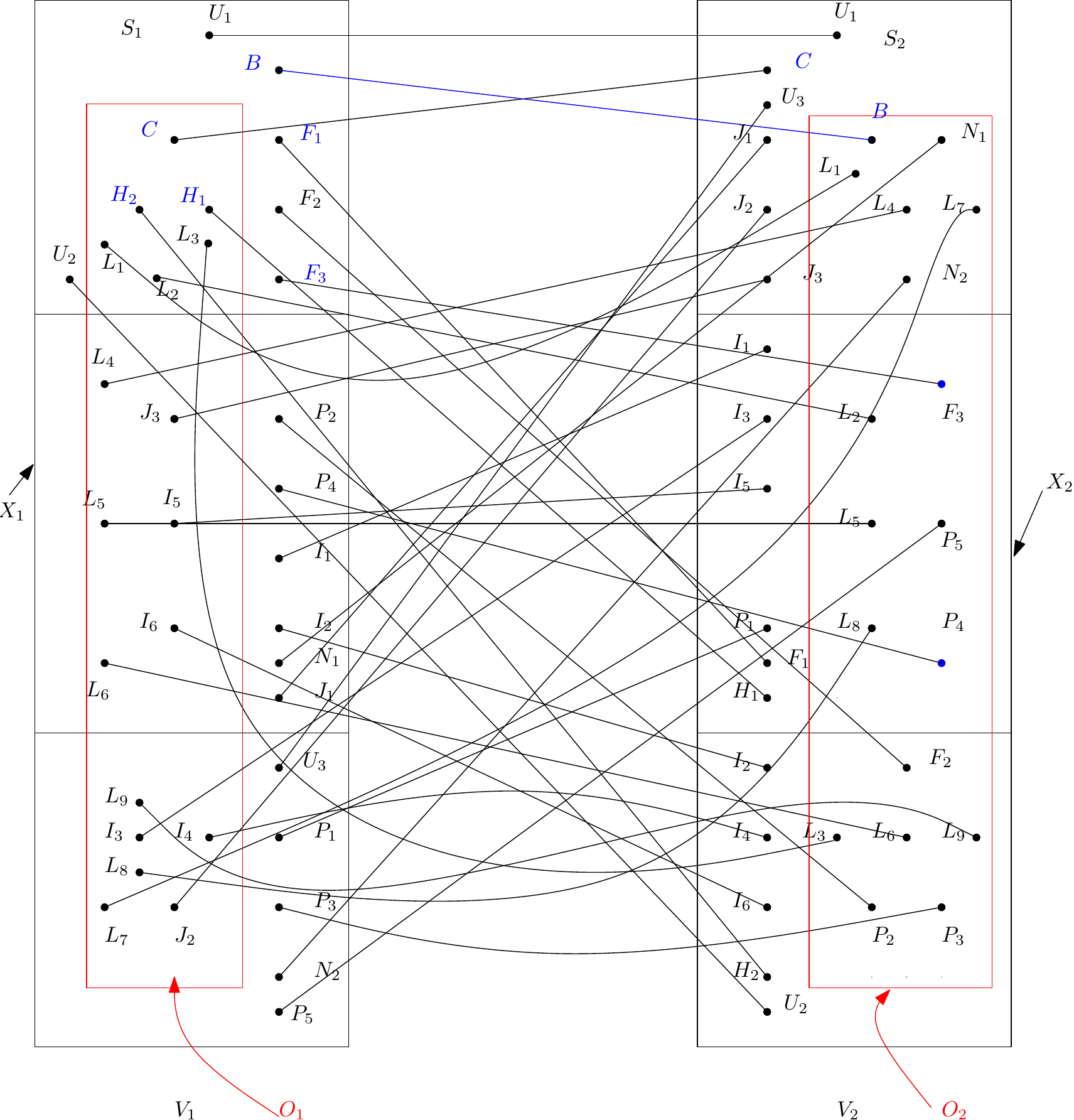}
\caption{Sets~$S_i$, $O_i$, $X_i$ $i= 1, 2$ and cuts between them.}
\end{center}
\end{figure}

\end{document}